\documentclass[12pt]{article}
\usepackage{enumerate}

\usepackage{amsmath,amsfonts,amsthm,amssymb,graphicx,natbib,geometry,arydshln,umoline,appendix,subfig,newtxtext,newtxmath,setspace}

\usepackage{color, comment}

\def \co {\mathbf{co}\,}

\def \ext {\mathbf{ext}\,}

\newtheorem{lemma}{Lemma}

\newtheorem{proposition}{Proposition}
\newtheorem{corollary}[proposition]{Corollary}
\newtheorem{claim}{Claim}

\theoremstyle{definition}
\newtheorem{definition}{Definition}
\newtheorem{example}{Example}

\theoremstyle{remark}
\newtheorem{remark}{Remark}

\begin{document}

\title{Robust Voting under Uncertainty\thanks{We are most grateful to the associate editor and two anonymous referees for detailed comments and suggestions, which have substantially improved this paper. 
We thank the seminar participants at the Workshop on Decision Making and Contest Theory (Ein Gedi, January 2017), Game Theory Workshop, Mathematical Economics at Keio, RUD, JEA Spring Meeting, SING, APET, Tokyo Metropolitan University, Hitotsubashi University, Tokyo University of Science, and Kwansei Gakuin University for their comments or for their suggestions. 
Ui acknowledges financial support by MEXT, Grant-in-Aid for Scientific Research (grant numbers 15K03348, 26245024, 18H05217). 
Nakada acknowledges financial support by MEXT, Grant-in-Aid for Scientific Research (grant numbers 16J03250, 19K13651).
Nitzan is grateful for Hitotsubashi  Institute  for  Advanced 
Study, Hitotsubashi University, for enabling the completion of this study.}}

\author{Satoshi Nakada\thanks{School of Management, Department of Business Economics, Tokyo University of Science; snakada@rs.tus.ac.jp.}\and  Shmuel Nitzan\thanks{Department of Economics, Bar-Ilan University; nitzans@biu.ac.il.}\and Takashi Ui\thanks{Hitotsubashi Institute for Advanced Study, Hitotsubashi University; oui@econ.hit-u.ac.jp.}}

\date{July 2025}
\maketitle

\begin{abstract}
\begin{spacing}{1.4}
This paper proposes normative criteria for voting rules under uncertainty about individual preferences. The criteria emphasize the importance of responsiveness, i.e., the probability that the social outcome coincides with the realized individual preferences. 
Given a convex set of probability distributions of preferences, denoted by $P$, 
a voting rule is said to be $P$-robust if, for each probability distribution in $P$,  at least one individual's responsiveness exceeds one-half. 
Our main result establishes that a voting rule is $P$-robust if and only if 
 there exists a nonnegative weight vector such that the weighted average of individual responsiveness is strictly greater than one-half under every extreme point of $P$. 
In particular, if the set $P$ includes all degenerate distributions, a $P$-robust rule is a weighted majority rule without ties.
\end{spacing}
\end{abstract}

\newpage

\section{Introduction}

Consider the choice of a voting rule on a succession of two alternatives (such as ``yes'' or ``no'') by a group of individuals. 
When a voting rule is chosen, the alternatives to come in the future are unknown, and the individuals are uncertain about their future preferences. 
An individual votes sincerely being concerned with the probability that the outcome agrees with his or her preference, which is referred to as responsiveness \citep{rae1969}. 
More specifically, an individual prefers a voting rule with higher responsiveness because he or she can expect that a favorable alternative is more likely to be chosen. 
For example, if an individual has a von Neumann-Morgenstern (VNM) utility function such that the utility from the passage of a favorable issue is one and that of an unfavorable issue is zero, then the expected utility equals the responsiveness. 

Imagine that a social planner proposes a voting rule under which the responsiveness of every individual is less than one-half. 
Then, the individuals may unanimously reject this proposal not only because the collective decision reflects minority preferences on average but also because there exists a voting rule that Pareto dominates this rule in terms of responsiveness. 
To demonstrate it, consider a voting rule whose collective decision always disagrees with that of the original rule, which is referred to as the inverse rule. 
Then, the responsiveness of every individual under the inverse rule is greater than one-half, i.e., greater than that under the original rule, because the sum of the former responsiveness and the latter responsiveness equals one. 
Therefore, the social planner regards it as a minimum requirement for a voting rule that the responsiveness of at least one individual should be greater than one-half, or equivalently, greater than the responsiveness under the inverse rule.

To find out whether the minimum  requirement is satisfied, the social planner must know the true probability distribution over preferences. 
In some circumstances, however, the social planner may not be able to identify the true distribution but instead has a set of possible distributions.
Then, the social planner asks the following question: Which voting rule satisfies this minimum requirement under each distribution in this set?
We call such a voting rule robustly undominated by the inverse rule, or robust for short,\footnote{We borrow the term ``robustness'' from robust statistics, statistics with good performance for data drawn from a wide range of probability distributions \citep{huber1981}.} and propose it as a normative criterion for voting rules. 
From the social planner's viewpoint, robustness is understood as a conservative criterion to avoid the worst-case scenario of violating the minimum requirement. 
In one sense, robustness may seem weak because it concerns the minimum requirement; in another sense, it appears strong, especially when the set includes all distributions, since the minimum requirement must be satisfied under all of them.

This paper introduces two types of robustness criteria and characterizes the classes of voting rules satisfying them for a given convex set of probability distributions of preferences, denoted by $P$. 
The criteria require that a voting rule should avoid the following worst-case scenarios. 
The first worst-case scenario is that the true responsiveness of every individual is less than or equal to one-half, or equivalently, the inverse rule weakly Pareto dominates this rule. 
A voting rule is said to be $P$-robust if, for any distribution in $P$,  it avoids this scenario.  
By replacing ``less than or equal to'' with ``strictly less than'' in the first worst-case scenario, we obtain a slightly more severe scenario. 
The second worst-case scenario is that the true responsiveness of every individual is strictly less than one-half, or equivalently, 
the inverse rule strictly Pareto dominates this rule. 
A voting rule is said to be weakly $P$-robust  if, for any distribution in $P$,  it avoids this scenario. 
Under weakly $P$-robust rules, the responsiveness of every individual can be less than or equal to one-half, as long as the responsiveness of at least one individual is equal to one-half. In this case, a collective decision is at best neutral to each individual's choice on average, which is never the case under $P$-robust rules. 

The main result establishes that a voting rule is $P$-robust if and only if it admits a nonnegative weight vector such that the weighted average of individual responsiveness is strictly greater than one-half under every extreme point of $P$. We also show that a voting rule is weakly $P$-robust if and only if the same weighted average is greater than or equal to one-half. 
In particular, when $P$ is the set of all distributions, a voting rule is $P$-robust if and only if it is a weighted majority rule (WMR) without ties, while a voting rule is weakly $P$-robust if and only if it is a WMR that allows ties with an arbitrary tie-breaking rule. 
In this case, $P$-robustness can be directly shown to be equivalent to efficiency within the set of all random voting rules, using the efficiency properties of voting rules discussed in \citet{azrielikim2016}.\footnote{We thank a reviewer for highlighting this issue, which is addressed in Section~\ref{Robustness vs. efficiency}.}

A concept closely related to $P$-robustness is belief-neutral efficiency introduced by \citet{brunnermeieretal2014}  in their study of financial markets. An allocation is said to be belief-neutral efficient if it is efficient with respect to every belief in the convex hull of individuals' subjective beliefs, which may be heterogeneous. 
When $P$ is the convex hull of all voters' subjective beliefs, $P$-robustness is equivalent to belief-neutral efficiency within the set of two voting rules, where each rule is the inverse rule of the other.

Both belief-neutral efficiency and $P$-robustness rely on a common set of beliefs, since Pareto dominance under heterogeneous beliefs is not normatively compelling due to the issue of spurious unanimity \citep{gilboaetal2004, mongin2016}.\footnote{\citet{blumeetal2018} compare complete markets and incomplete markets when individuals have heterogeneous beliefs and demonstrate potential social benefits from restrictions on trade that make markets incomplete.} 
To illustrate spurious unanimity in a voting context, consider individuals voting on a proposed amendment. There are two types of individuals: Type 1 individuals prefer approval and believe that more than two-thirds of voters share their preference; Type 2 individuals prefer disapproval and believe that more than two-thirds of voters share their preference. When a two-thirds rule is proposed, both types unanimously support its adoption, as each believes it will lead to their favored outcome.
However, as emphasized by \citet{mongin2016}, such unanimity is spurious: it arises from mutually inconsistent beliefs.

The rest of this paper is organized as follows. 
Section~\ref{A weighted majority rule} provides a summary of the notation.  
Section~\ref{Voting under Knightian uncertainty} introduces the concepts of $P$-robustness. 
Section~\ref{Robustness of nonanonymous rules} provides the main results. 
Section~\ref{Robustness vs. efficiency} discusses efficient voting rules in comparison to $P$-robust voting rules.   
We discuss two extensions in Section~\ref{section: discussions} and conclude the paper in Section \ref{Conclusion}.

\subsection{Related literature}\label{Related literature}

\citet{rae1969} and \citet{taylor1969} were the first to use responsiveness to characterize voting rules, followed by \citet{straffin1977} and \citet{fleurbaey2008}. 
Their result, which we call the Rae-Taylor-Fleurbaey (RTF) theorem,  states that a voting rule is a WMR if and only if it maximizes the corresponding weighted sum of responsiveness over all individuals, where a probability distribution of preferences is assumed to be known. The normative implication of the RTF theorem is efficiency of WMRs:\footnote{See also \citet{schmitztroger2012} and \citet{azrielikim2014}.} a voting rule is efficient in the set of random voting rules in terms of responsiveness if and only if it is a WMR.

There are other important strands of literature on voting rules using responsiveness.
\citet{barberajackson2004} ask which voting rule can survive in the long run when individuals can propose to replace the current rule with a different one, where preferences over voting rules are given by responsiveness. 
A voting rule is said to be self-stable if no other rule can replace it with sufficient support from the group, where the decision to change the rule is made by applying the original rule itself. Focusing on anonymous rules, \citet{barberajackson2004} show that self-stable rules may not exist, and also establish conditions that guarantee their existence. 
\citet{azrielikim2016} extend the analysis to nonanonymous rules and identify self-stable WMRs. 

\citet{schmitztroger2012} and \citet{azrielikim2014} consider a direct two-alternative mechanism, where individuals report their cardinal utility, and relate it to SMRs and WMRs. %\footnote{These papers consider WMRs with quotas. See also Footnote 4.}    
Their approach is in clear contrast to that of the above papers as well as ours, where individuals report their ordinal ranking of two alternatives. 
\citet{schmitztroger2012} focus on a symmetric environment and show that a SMR is the utilitarian rule subject to incentive compatibility constraints and that it is interim incentive efficient. 
\citet{azrielikim2014} extend the analysis to asymmetric environments and show that the utilitarian rule subject to incentive compatibility constraints is a WMR, and a rule is interim incentive efficient if and only if it is a WMR.

This paper also joins a literature on the worst-case approach to economic design. 
Most studies in this literature focus on design of a mechanism maximizing the worst-case value of an objective function. 
For example, \citet{chungely2007} consider a revenue maximization problem in a private value auction where the auctioneer does not know the agents' exact belief structures\footnote{This is a standard assumption in robust mechanism design \citep{bergemannmorris2005}.} and show that the optimal auction rule is a dominant-strategy mechanism when the auctioneer evaluates rules by their worst-case performance.  
\citet{carroll2015} considers a moral hazard problem where the principal does not know the agent's set of  possible actions exactly and shows that the optimal contract is linear when the principal evaluates contracts by their worst-case performance.\footnote{Other recent examples of economic design with worst-case objectives include \citet{bergemannschlag2011}, \citet{yamashita2015}, \citet{carroll2017}, 
\citet{chenli2017}, \citet{auster2018}, \citet{carrascoet al2018}, \citet{du2018}, \citet{bergemannetal2019}, \citet{yamashitazhu2021}, and \citet{brooksdu2021,brooksdu2024} among others. See also a survey by \citet{carroll2018}. In decision theory, \citet{gilboaschmeidler1989} is a seminal paper.} 
The maximin criterion adopted in the aforementioned papers requires ``performance robustness,'' whereas our criterion is based on ``Pareto-undominance robustness,'' a version of the Pareto criterion evaluated under multiple beliefs that also avoids the problem of spurious unanimity. 

\section{Voting rules}\label{A weighted majority rule}
Consider a group of individuals $N=\{1,\ldots, n\}$ that faces a choice between two alternatives $-1$ and $1$.  
The choice of individual $i\in N$ is represented by a decision variable $x_i\in\{-1,1\}$. 
%We assume sincere voting except in Section \ref{Robust voting as a mechanism}. 
The choices of the group members are summarized by a decision profile $x=(x_i)_{i\in N}$. Let $\mathcal{X}= \{-1,1\}^n$ denote the set of all possible decision profiles. 

A deterministic voting rule, or a voting rule for short, is a mapping $\phi: \mathcal{X}\to\{-1,1\}$, which assigns a collective decision $\phi(x)\in \{-1,1\}$ to each $x\in \mathcal{X}$. 
The set of all voting rules is denoted by $\Phi$. 
For $\phi\in \Phi$, 
let $-\phi\in \Phi$ be the voting rule whose collective decision always disagrees with that of $\phi$, i.e., $(-\phi)(x)=-\phi(x)$ for all $x\in \mathcal{X}$. 
This voting rule is referred to as the inverse rule of $\phi$.  
Although we mainly study a deterministic voting rule, 
we occasionally consider a random voting rule as well. 
A random voting rule is a mapping $\bar\phi:\mathcal{X}\to [-1,1]$, which is interpreted as follows: $\bar \phi$ assigns a collective decision $+1$ with probability $(1+\bar\phi(x))/2$ and $-1$ with probability $(1-\bar\phi(x))/2$ to each $x\in \mathcal{X}$. 
The set of all random voting rules is denoted by $\Delta(\Phi)$.

A voting rule $\phi\in \Phi$ is a weighted majority rule  (WMR)\footnote{For an axiomatic foundation of weighted majority rules, see \citet{fishburn1973}, \citet{fishburngehrlein1977}, \citet{nitzanparoush1981}, \citet{einylehrer1989}, and \citet{taylorzwicker1992} among others.} if there exists a nonzero weight vector\footnote{We allow negative weights, which appear in Proposition \ref{Fleurbaey}.}  $w=(w_i)_{i\in N}\in \mathbb{R}^n$ satisfying
\begin{equation}
\phi(x)=
\begin{cases}
1 & \text{ if }\sum_{i\in N}w_ix_i> 0,\\
-1 & \text{ if }\sum_{i\in N}w_ix_i<0.
%\end{cases}
\end{cases} \label{def: SMR}
\end{equation}
A simple majority rule (SMR) is a special case with positive equal weights, i.e., $w_i=w_j>0$ for all $i,j\in N$. 
When there is a tie, i.e. $\sum_{i\in N}w_ix_i= 0$, a tie-breaking rule is used to determine a WMR. 
For example, a SMR requires a tie-breaking rule if $n$ is even. 
A tie-breaking rule can be either deterministic or random. 
A WMR with a random tie-breaking rule is a random voting rule $\bar\phi\in \Delta(\Phi)$ satisfying (\ref{def: SMR}) (with $\phi$ replaced by $\bar \phi$).

A voting rule $\phi\in \Phi$ is anonymous if it is symmetric in its $n$ variables; that is, $\phi(x)=\phi(x^\pi)$ for each $x\in \mathcal{X}$ and each permutation $\pi:N\to N$, where $x^\pi=(x_{\pi(i)})_{i\in N}$.  A SMR is anonymous if $n$ is odd or if $n$ is even and its tie-breaking rule is anonymous, i.e., symmetric in its $n$ variables. A WMR with nonnegative weights is anonymous if and only if it is an anonymous SMR. %\footnote{To see why the ``only if'' part is true, suppose that an anonymous WMR with nonnegative weights is not a SMR. Then, there exists $S\subset N$ such that $|S|<n/2$ and $\sum_{i\in S}w_i>\sum_{i\in N\setminus S}w_i$. By anonymity, for $S'\subset N\setminus S$ with $|S'|=|S|$, we have $\sum_{i\in S'}w_i>\sum_{i\in N\setminus S'}w_i$. Thus,  $\sum_{i\in S}w_i+\sum_{i\in S'}w_i>\sum_{i\in N\setminus S}w_i+\sum_{i\in N\setminus S'}w_i$, which is a contradiction because $S\subseteq N\setminus S'$ and $S'\subseteq N\setminus S$. } 

The following characterization of WMRs is immediate from the definition.
\begin{lemma}\label{weighted majority characterization}
A voting rule $\phi\in\Phi$ is a WMR with a weight vector $w\in \mathbb{R}^n$ 
if and only if 
\begin{equation}
\phi(x)\sum_{i\in N}w_ix_i\geq 0 \, \text{ for all }\, x\in\mathcal{X}. \label{weighted majority}
\end{equation}
A voting rule $\phi\in \Phi$ is a WMR with a weight vector $w\in \mathbb{R}^n$ allowing no ties if and only if 
\begin{equation}
\phi(x)\sum_{i\in N}w_ix_i> 0 \, \text{ for all }\, x\in\mathcal{X}. \label{weighted majority 2}
\end{equation}
\end{lemma}

\section{Voting under uncertainty}\label{Voting under Knightian uncertainty}

Assume that $x\in \mathcal{X}$ is randomly drawn according to a probability distribution $p\in \Delta(\mathcal{X})\equiv\{p'\in \mathbb{R}_+^{|\mathcal{X}|}:\sum_{x\in \mathcal{X}}p'(x)=1\}$. 
For a deterministic voting rule $\phi\in \Phi$, let 
\[
r_i(\phi,p)\equiv p(\{x\in \mathcal{X}:\phi(x)=x_i\})=\sum_{x\in \mathcal{X}:\phi(x)=x_i}p(x)
\]
be the probability that $i$'s choice agrees with the collective decision, which is referred to as responsiveness or the Rae index \citep{rae1969}. 
It is calculated as 
\begin{equation}
r_i(\phi,p)=(E_p[\phi(x)x_i]+1)/2
\label{fundamental 2}
\end{equation}
because 
\begin{equation}
E_p[\phi(x)x_i]%=\sum_{x\in \mathcal{X}}p(x)\phi(x)x_i
=\sum_{x:\phi(x)=x_i}p(x)-\sum_{x:\phi(x)\neq x_i}p(x)
=2r_i(\phi,p)-1. \notag\label{fundamental 3}
\end{equation}
Using the above calculation and Lemma \ref{weighted majority characterization}, 
we can characterize {WMRs with nonnegative weights as the class of voting rules} for which the weighted average of individual responsiveness exceeds one-half under every probability distribution.
\begin{lemma}\label{weighted majority characterization responsiveness}
A voting rule $\phi\in\Phi$ is a WMR with a weight vector $w\in \mathbb{R}_+^n$ 
if and only if 
\begin{equation}
\frac{\sum w_ir_i(\phi,p)}{\sum w_i}\geq \frac{1}{2} \ \text{ for all }\ p\in\Delta(\mathcal{X}). \label{weighted majority res}
\end{equation}
A voting rule $\phi\in \Phi$ is a WMR with a weight vector $w\in \mathbb{R}_+^n$ allowing no ties if and only if 
\begin{equation}
\frac{\sum w_ir_i(\phi,p)}{\sum w_i}> \frac{1}{2} \ \text{ for all }\  p\in\Delta(\mathcal{X}). \label{weighted majority res 2}
\end{equation}
\end{lemma}
\begin{proof}
Taking the expectation on both sides of \eqref{weighted majority} and applying \eqref{fundamental 2} yields \eqref{weighted majority res}.
Conversely, \eqref{weighted majority} follows from \eqref{weighted majority res} by considering each degenerate distribution that assigns probability one to a single $x \in \mathcal{X}$.
Therefore, \eqref{weighted majority} and \eqref{weighted majority res} are equivalent.
An analogous argument shows that \eqref{weighted majority 2} is equivalent to \eqref{weighted majority res 2}.
\end{proof}

An individual is assumed to vote sincerely and prefer a voting rule with higher responsiveness because he or she expects that a favorable alternative is more likely to be chosen. A discussion and justification of these assumptions appear in the online appendix.\footnote{We give two alternative definitions of robustness. Redefining robustness as a requirement for a voting
mechanism, where individuals can vote insincerely, we show that sincere voting is weakly dominant for every
individual if the mechanism satisfies robustness. Redefining robustness using the expected utility of individuals,
we provide justification for the original definition based on responsiveness.}

We consider the following Pareto order over $\Phi$. 
We say that a voting rule $\phi\in \Phi$ is weakly Pareto-preferred to $\phi'\in \Phi$ %(in terms of responsiveness) 
under $p$ if $r_i(\phi,p)\geq r_i(\phi',p)$ for all $i\in N$. 
Similarly, we say that $\phi$ is Pareto-preferred to $\phi'$ under $p$ if $r_i(\phi,p)\geq r_i(\phi',p)$ for all $i\in N$ with strict inequality holding for at least one individual, and that $\phi$ is strictly Pareto-preferred to $\phi'$ under $p$ if $r_i(\phi,p)> r_i(\phi',p)$ for all $i\in N$. 
%This terminology also applies to random voting rules. 

Imagine that a social planner proposes a voting rule $\phi\in \Phi$ such that the responsiveness of every individual is less than or equal to one-half, 
\begin{equation}
r_i(\phi,p)\leq 1/2 \text{ for all }i\in N, \label{robustness condition 1}
\end{equation}
which means that the collective decision under $\phi$ reflects minority preferences on average.
Note that \eqref{robustness condition 1} is true if and only if the inverse rule $-\phi$ is weakly Pareto-preferred to $\phi$,  
\begin{equation}
r_i(-\phi,p)\geq r_i(\phi,p) \text{ for all }i\in N,\notag\label{robustness condition 1'}
\end{equation} 
because 
$r_i(\phi,p)+ r_i(-\phi,p)=1$. 
Thus, the social planner regards it as a minimum requirement for a voting rule that the responsiveness of at least one individual should be greater than one-half. 

Let $P\subset \Delta(\mathcal{X})$ be a convex set of possible probability distributions of preferences. The social planner knows that the true distribution lies in $P$, but cannot identify it and therefore cannot accurately compute responsiveness.
Under these circumstances, the collective worst-case scenario for adopting $\phi$ is that the true responsiveness of every individual is less than or equal to one-half, or equivalently, the inverse rule $-\phi$ is weakly Pareto-preferred to $\phi$ under the true distribution. 
We say that $\phi$ is robustly undominated by the inverse rule under $P$, or $P$-robust for short, if this worst-case scenario is avoided under every distribution in $P$. 
Note that, if $P \subset P' \subset \Delta(\mathcal{X})$, then $P$-robustness is a weaker requirement than $P'$-robustness by definition.

\begin{definition}\label{def: robust}
A deterministic voting rule $\phi\in\Phi$ is {\em $P$-robust} if the inverse rule $-\phi$ is not weakly Pareto-preferred to $\phi$ under any $p\in P$, or equivalently, for each  $p\in P$, the responsiveness of at least one individual is strictly greater than one-half. A $\Delta(\mathcal{X})$-robust rule is simply referred to as a robust rule. 
\end{definition}

An immediate sufficient condition for $P$-robustness is that there exists a nonnegative weight vector such that the weighted average of individual responsiveness is strictly greater than one-half for each $p \in P$.
If a voting rule satisfies this condition, then the maximum responsiveness must also be strictly greater than one-half, implying $P$-robustness.
In particular, by Lemma~\ref{weighted majority characterization responsiveness}, any WMR with nonnegative weights and no ties satisfies this sufficient condition for arbitrary $P$.
Our main results show that the same condition is also necessary, as a consequence of a theorem of the alternative.

%This criterion is similar to, but different from, the belief-neutral Pareto criterion introduced by \citet{brunnermeieretal2014} in their study of financial markets. To translate it into our model, assume that individuals' subjective beliefs $p_1,\ldots,p_n\in \Delta(\mathcal{X})$ are observable to a social planner and the set of reasonable beliefs is the convex hull of them, i.e., $P=\co \{p_1,\ldots,p_n\}$. Under this assumption, $\phi$ is said to be belief-neutral efficient if it is efficient within $\{\phi, -\phi\}$ under any $p\in P$. 

We also consider a weaker version of $P$-robustness by considering the following more severe collective worst-case scenario for adopting $\phi$: the true responsiveness of every individual is {\em strictly} less than one-half, or equivalently, the inverse rule $-\phi$ is {\em strictly} Pareto-preferred to $\phi$ under the true distribution.  
We say that $\phi$ is weakly $P$-robust if this worst-case scenario is avoided under every distribution in $P$.

%A voting rule is said to be weakly robust if it avoids this scenario, which is a weaker requirement than robustness. 

\begin{definition}\label{def: weak robust}
A deterministic voting rule $\phi\in\Phi$ is {\em weakly $P$-robust} if the inverse rule $-\phi$ is not strictly Pareto-preferred to $\phi$ under any $p\in P$,  or equivalently, for each  $p\in P$, the responsiveness of at least one individual is greater than or equal to one-half. 
A weakly $\Delta(\mathcal{X})$-robust rule is simply referred to as a weakly robust rule. 
%$p(\phi(x)=x_i)\geq 1/2$ for at least one $i\in N$. 
\end{definition}

If $\phi\in \Phi$ is weakly $P$-robust but not $P$-robust, there exists $p\in P$ such that 
$\max_{i\in N}r_i(\phi,p)= 1/2$. 
%Thus, the responsiveness of every individual is less than or equal to one-half. 
If such $p$ is the true distribution, a collective decision is at best neutral to each individual's choice on average, 
which is never the case with robust rules.

%Hereafter, we refer to a $\Delta(\mathcal{X})$-robust rule and a weakly $\Delta(\mathcal{X})$-robust rule simply as a robust rule and a weakly robust rule, respectively.

To illustrate the concept of $P$-robustness, we present two examples. 
In both cases, for a given voting rule $\phi$, we denote by $Q_\phi\subset\Delta(\mathcal{X})$ the set of probability distributions under which the arithmetic mean of individual responsiveness is strictly greater than one-half.

\begin{example}
Consider the two-thirds rule denoted by $\phi$ with $9$ individuals: $\phi(x)=1$ if and only if $\#\{i : x_i = 1\} \geq 6$. 
Then, for arbitrary $P\subseteq Q_\phi$, $\phi$ is $P$-robust since at least one individual's responsiveness is strictly greater than one-half under every $p\in Q_\phi$. 
Compare the two-thirds rule and the SMR and observe that both yield identical outcomes except when five individuals choose $1$ and four choose $-1$.
Thus, the total responsiveness under the SMR exceeds that under $\phi$ by $p(\{x : \#\{i : x_i = 1\} = 5\})$ for each $p \in Q_\phi$ (see \eqref{total responsiveness}). 
This implies that the two-thirds rule $\phi$ is Pareto-dominated by the SMR under $p \in Q_\phi$ such that $p(\{x : \#\{i : x_i = 1\} = 5\}) > 0$, provided the symmetry condition holds: $p(x) = p(x^\pi)$ for all $x \in \mathcal{X}$ and $\pi \in \Pi$.\footnote{For example, such $p$ is given by $p(x)=1-\varepsilon$ if $\#\{i:x_i=1\}=9$ and $p(x)=\varepsilon/\binom{9}{5}$ if $\#\{i:x_i=1\}=5$, where $0<\varepsilon<1/2$.}
In other words, $Q_\phi$-robustness does not imply efficiency.
\end{example}

\begin{example}
Consider the unanimity rule denoted by $\phi$ with $n \geq 3$, where the collective decision is $1$ if and only if all individuals vote $1$. 
This rule is $P$-robust for any $P \subseteq Q_\phi$, but it can also be $P$-robust for some $P$ with $P \cap Q_\phi = \emptyset$. 
For instance, let $p_i \in \Delta(\mathcal{X})$ be the degenerate distribution that assigns probability one to the profile in which individual $i$ votes $-1$ while all others vote~$1$.
Then, $\phi$ is $\{p_i\}$-robust for each $i\in N$ since the responsiveness of individual $i$ is one. 
Thus, although $P = \{p_i\}_{i \in N}$ is not convex, $\phi$ satisfies the requirement of $P$-robustness. Note that $P\cap Q_\phi=\emptyset$. 
On the other hand, $\phi$ is not $\co(P)$-robust, where $\co(P)$ denotes the convex hull of $P$.
To see this, let $p = \sum_{i \in N} p_i / n \in \co(P)$. 
Then,
\[
r_i(\phi, p) = \sum_{j \in N} r_i(\phi, p_j)/n = 1/n < 1/2
\]
for each $i \in N$, implying that $\phi$ is not $\co(P)$-robust.
\end{example}

\begin{comment}

\begin{example}\label{example 1}
Assume that $n=5$ and let $\phi$ be a SMR with veto power of individual $1$: 
% \[
%\phi(x)=
%\begin{cases}
%1 & \text{ if $x_1=1$ and $\#\{i: x_i=1\}\geq 3$},\\
%-1 & \text{otherwise}.%\text{ if $x_1=-1$ or $\#\{i: i\geq 2, x_i=-1\}\geq n-2$}.
%\end{cases}
%\]
\begin{equation}
\phi(x)=
\begin{cases}
-x_1 & \text{ if $x\in\mathcal{A}$},\\
x_1 & \text{ otherwise},
\end{cases}\label{example voting 1}
\end{equation}
where $\mathcal{A}\equiv \{x\in \mathcal{X}:x_1=1,\ \#\{i: x_i=1\}\leq 2\}$. 
Then, $\phi$ is not weakly robust because the inverse rule $-\phi$ is strictly Pareto-preferred to $\phi$ under $p\in \Delta(\mathcal{X})$ given by 
\[
p(x)=
\begin{cases}
	1/16 & \text{if $x=(1,-1,-1,-1,-1)$},\\
	2/16 & \text{if $x\in 
	\{(1,1,-1,-1,-1),(1,-1,1,-1,-1),(1,-1,-1,1,-1),(1,-1,-1,-1,1)\}
	$},\\
	%\text{if $x_1=1$ and $\#\{i: x_i=1\}=2$},\\
	7/16 & \text{if $x=(-1,1,1,1,1)$},\\
	0 & \text{otherwise}.
\end{cases}
\]
In fact, 
$E_{p}[\phi(x)x_i]=7/16-9/16=-1/8<0$ for all $i$,  and thus 
$E_{p}[-\phi(x)x_i]=-E_{p}[\phi(x)x_i]>E_{p}[\phi(x)x_i]$ for all $i$. This implies that $\phi$ does not satisfy the requirement of weak robustness as well as that of robustness by (\ref{fundamental 2}). 
\end{example}

\end{comment}

\section{Main results}\label{Robustness of nonanonymous rules}

\subsection{Characterizations}\label{Characterizations}

In the main results that follow, we show that a voting rule is $P$-robust if and only if there exists a nonnegative weight vector such that the weighted average of individual responsiveness $r_i(\phi,p)$ is strictly greater than one-half under every extreme point of $P$. 
We also show that a voting rule is weakly $P$-robust if and only if the same weighted average is greater than or equal to one-half.

\begin{proposition}\label{main result}
Suppose that $P\subseteq \Delta(\mathcal{X})$ is the convex hull of a finite set $\{p_j\}_{j\in M}\subsetneq \Delta(\mathcal{X})$, where $M=\{1,\ldots, m\}$. 
Then, the following results hold. 
\begin{enumerate}[\em (1)]
	\item 
A voting rule $\phi\in \Phi$ is $P$-robust if and only if there exists $w\in \mathbb{R}_{+}^n$ such that 
\begin{equation}
\frac{\sum w_i r_i(\phi,p_j)}{\sum w_i}> 1/2\, \text{ for all }\, j\in M.
\label{robustness condition w1}    	
\end{equation}
In particular, a deterministic voting rule is robust if and only if it is a WMR with nonnegative weights such that there are no ties. 

\item A voting rule $\phi\in \Phi$ is weakly $P$-robust if and only if there exists $w\in \mathbb{R}_{+}^n$ such that 
\begin{equation}
\frac{\sum w_i r_i(\phi,p_j)}{\sum w_i}\geq 1/2\, \text{ for all }\, j\in M.
\label{robustness condition w2}    	
\end{equation}
In particular, a deterministic voting rule is weakly robust if and only if it is a WMR with nonnegative weights. 
\end{enumerate}
\end{proposition}

\begin{remark}
Proposition \ref{main result} assumes that $P$ is the convex hull of a finite set $\{p_j\}_{j\in M}$, but we can obtain the same result even if $P$ is an arbitrary closed convex set. See Appendix \ref{an arbitrary convex set}.
\end{remark}

\begin{remark}
Since responsiveness is a linear function of a probability distribution, $P$-robustness implies that \eqref{robustness condition w1} holds not only for each $p_j$ but for all $p \in P$.
Similarly, weak $P$-robustness implies that \eqref{robustness condition w2} holds for all $p \in P$, not just for each $p_j$.
\end{remark}

In the case of anonymous rules, it is reasonable to consider $P$-robustness with a structure that is invariant under permutations of individuals' indices, thereby respecting anonymity. 
We say that $P\subset\Delta(\mathcal{X})$ is permutation invariant if $p\in P$ implies $p^\pi\in P$ for any permutation $\pi$, where $p^\pi$ is defined by $p^\pi(x)=p(x^\pi)$ for all $x\in \mathcal{X}$. 
The following corollary characterizes $P$-robust anonymous rules with permutation invariant $P$, where we can choose equal weights $w_1=\cdots=w_n$. 
\begin{corollary}\label{main result cor new}
In addition to the conditions in Proposition \ref{main result}, 
assume that $P$ is permutation invariant. 
Then, 
an anonymous voting rule $\phi\in \Phi$ is $P$-robust if and only if the arithmetic mean of individual responsiveness is strictly greater than one-half under every $p_j$: 
\begin{equation}
\sum_{i\in N} r_i(\phi,p_j)/n> 1/2\,\text{ for all }\, j\in M.
\label{robustness condition r1}    	
\end{equation}
Similarly, an anonymous voting rule $\phi\in \Phi$ is weakly $P$-robust if and only if  the arithmetic mean of individual responsiveness is greater than or equal to one-half under every $p_j$: 
\begin{equation}
\sum_{i\in N} r_i(\phi,p_j)/n\geq 1/2\, \text{ for all }\, j\in M.
\label{robustness condition r2}    	
\end{equation}
\end{corollary}

Corollary \ref{main result cor new} implies that, for any permutation invariant convex set $P$ (not necessarily the convex hull of a finite set), if $\phi$ is $P$-robust, then under every $p \in P$, the average responsiveness is strictly greater than one-half. 
To see why, fix an arbitrary $p \in P$, and let $\co(\{p^\pi\})$ denote the convex hull of all permutations $\{p^\pi\}$ of $p$. By the permutation invariance and convexity of $P$, we have $\co(\{p^\pi\}) \subset P$, so $\phi$ is also $\co(\{p^\pi\})$-robust. It then follows from Corollary \ref{main result cor new} that \eqref{robustness condition r1} holds for $p_j = p$.

To apply this observation to Examples 1 and 2, we express the arithmetic mean of individual responsiveness in terms of a probability distribution.  
For an anonymous rule $\phi$ and a probability distribution $p$, 
let $p_k$ and $d_{k,\phi}$ denote the probability that exactly $k$ individuals choose $1$, and the number of individuals whose choice coincides with the collective decision when exactly $k$ individuals choose $1$ under $\phi$, respectively. 
Then, the total responsiveness is given by 
$
\sum_{k=0}^nd_{k,\phi}\cdot p_k$ since 
\begin{equation}
\sum_{i \in N}r_i(\phi,p)=
\sum_{i \in N}\sum_{x:\phi(x)=x_i}p(x)=\sum_{x\in\mathcal{X}}\#\{i \in N: x_i=\phi(x)\}\cdot p(x)=\sum_{k=0}^nd_{k,\phi}\cdot p_k,\label{total responsiveness}
\end{equation}
and dividing this by $n$ yields the average responsiveness.
%Thus, by Corollary \ref{main result cor new}, an anonymous rule $\phi$ is $P$-robust for symmetric $P$ if and only if each $p\in P$ satisfies$\sum_{k=0}^nd_{k,\phi} p_k/n>1/2$. 

\setcounter{example}{0}

\begin{example}[continued]
The two-thirds rule is an anonymous rule.
Let $P\subset\Delta(\mathcal{X})$ be a permutation invariant closed convex set. Then, 
Corollary \ref{main result cor new} implies that $\phi$ is $P$-robust if and only if $P\subset Q_\phi$, where $Q_\phi$ is the set of probability distributions satisfying 
\[
\sum_{k=0}^5 (9-k)\cdot p_k
+\sum_{k=6}^9 k\cdot p_k>9/2.
\]
\end{example}

\begin{example}[continued]
The unanimity rule is an anonymous rule. 
Recall that $\phi$ is $\{p_i\}$-robust for each $i\in N$ and note that $\{p_i\}$ is not permutation invariant. 
On the other hand, $\phi$ is $P$-robust for a permutation invariant closed convex set $P \subset \Delta(\mathcal{X})$ if and only if $P \subset Q_\phi$, where $Q_\phi$ is the set of probability distributions satisfying
\[
\sum_{k=0}^{n-1} (n-k)\cdot p_k + n\cdot p_n > \frac{n}{2}.
\]
\end{example}

\subsection{Proof of the characterizations}

To sketch the proof for Part (1) of Proposition \ref{main result}, recall that $P$ is the convex hull of $\{p_j\}_{j\in M}$. 
Thus, $\phi$ is {\em not} $P$-robust if and only if there exists $\lambda\in \Delta(M)$ such that 
\begin{equation}
r_i(\phi,\sum_{j\in M}\lambda_jp_j)=\sum_{j\in M}\lambda_jr_i(\phi,p_j)\leq 1/2 \text{ for all }i\in N, 
\notag
\end{equation}
which is equivalent to
\begin{equation}
\sum_{j\in M}\lambda_jE_{p_j}[\phi(x)x_i]
\leq 0 \text{ for all }i\in N \label{robustness condition 2}    
\end{equation}
 by \eqref{fundamental 2}. 
 Using a theorem of alternatives, 
 we can obtain a necessary and sufficient condition for the existence of a solution to \eqref{robustness condition 2}. The condition is equivalent to the nonexistence of a weight vector $w\in\mathbb{R}_+^n$ that satisfies 
\begin{equation}
    \sum_{i\in N}w_iE_{p_j}[\phi(x) x_i]  
>0\text{ for all }j\in M, \notag\label{robustness condition 2'}    
\end{equation}
which is rewritten as \eqref{robustness condition w1} by \eqref{fundamental 2}.
We can prove Part (2) similarly.

For the formal proofs, we use the following inequality symbols.  For vectors  $\xi$ and $\eta$, we write $\xi\geq \eta$ if $\xi_i\geq \eta_i$ for each $i$, $\xi>\eta$ if $\xi_i\geq \eta_i$ for each $i$ and $\xi\neq \eta$, and $\xi\gg \eta$ if $\xi_i>\eta_i$ for each $i$.

We  write $l_{ij}=E_{p_j}[\phi(x)x_i]$ for each $(i,j)\in N\times M$ and define the $n\times m$ matrix $L=[l_{ij}]_{n\times m}$. 
Part (1) states that exactly one of the following holds. 
\begin{enumerate}[(a)]
\item 
There exists $w=(w_i)_{i\in N}\geq 0$ such that
$
\sum_{i\in N} w_i l_{ij}=
\sum_{i\in N}w_iE_{p_j}[\phi(x) x_i]  
>0
$
for each $j\in M$, 
or equivalently, 
$w^\top L\gg 0$.
\item There exists $\lambda=(\lambda_j)_{j\in M}> 0$  such that 
$
\sum_{j\in M} l_{ij}\lambda_j
=\sum_{j\in M}\lambda_jE_{p_j}[\phi(x)x_i]
\leq 0
$
for each $i\in N$, or equivalently,  
$L\lambda\leq 0$.
\end{enumerate}
Similarly, 
Part (2) states that exactly one of the following holds. 
\begin{enumerate}[(a$^\prime$)]
\item 
There exists $w=(w_i)_{i\in N}> 0$ such that
$
\sum_{i\in N} w_i l_{ij}=
\sum_{i\in N}w_iE_{p_j}[\phi(x) x_i]  
\geq 0$ 
for each $j\in M$, 
or equivalently, 
$w^\top L\geq 0$.
\item There exists $\lambda=(\lambda_j)_{j\in M}> 0$  such that 
$
\sum_{j\in M} l_{ij}\lambda_j
=\sum_{j\in M}\lambda_jE_{p_j}[\phi(x)x_i]
< 0
$
for each $i\in N$, or equivalently,  
$L\lambda\ll 0$.
\end{enumerate}

The following theorem of alternatives due to \citet{neumannmorgenstern1944}\footnote{\citet{neumannmorgenstern1944} use this result to prove the minimax theorem.} guarantees that 
exactly one of (a) and (b) holds, and that exactly one of (a$^\prime$) and (b$^\prime$) holds. The same result also appears in the work of \citet[Theorem 2.10]{gale1960} as a corollary of Farkas' lemma. 

\begin{lemma}\label{alternative1}
Let $A$ be an $n\times m$ matrix. Exactly one of the following alternatives holds. 
\begin{itemize}
\item There exists $\xi\in \mathbb{R}^n$ satisfying 
\[
\xi^\top A \gg 0, \ \xi \geq 0.
\]
\item There exists $\eta\in \mathbb{R}^m$ satisfying 
\[
 A\eta\leq 0,\ \eta>0.
\]
\end{itemize}
\end{lemma}

\begin{proof}[Proof of Proposition \ref{main result}]
To establish Part (1), plug $L$, $w$, and $\lambda$ into  $A$, $\xi$, and $\eta$ in Lemma \ref{alternative1}, respectively. 
Then, Lemma \ref{alternative1} implies that exactly one of (a) and (b) holds. 

To establish Part (2), plug 
 $-L^\top$, $\lambda$, and $w$ into  $A$, $\xi$, and $\eta$ in Lemma \ref{alternative1}, respectively, where we replace $(n,m)$ with $(m,n)$. 
Then, Lemma \ref{alternative1} implies that exactly one of (a$^\prime$) and (b$^\prime$) holds.  
\end{proof}

\begin{proof}[Proof of Corollary \ref{main result cor new}]
For each statement, it suffices to show the ``only if'' part. 
Suppose that an anonymous rule $\phi$ is $P$-robust. 
Since $P$ is permutation invariant, we must have $\{p_j\}_{j\in M}=\{p_j^\pi\}_{j\in M}$ for each permutation $\pi$. 
By Proposition \ref{main result},  there exists $w\in \mathbb{R}_{+}^n$ satisfying \eqref{robustness condition w1}.    Then, for any permutation $\pi$, it holds that 
\begin{equation}
\frac{\sum_{i}w_{i} r_i(\phi,p_j^\pi)}{\sum_{i}w_i}=\frac{\sum_{i}w_{\pi(i)} r_i(\phi,p_j)}{\sum_{i}w_i}> 1/2\text{ for all }j\in M
\label{cor key eq new}
\end{equation} 
since 
\[
r_i(\phi,p^\pi)=\sum_{x:\phi(x)=x_i}p(x^\pi)
=\sum_{x:\phi(x^{\pi^{-1}})=x_{\pi^{-1}(i)}}p(x)
=\sum_{x:\phi(x)=x_{\pi^{-1}(i)}}p(x)=r_{\pi^{-1}(i)}(\phi,p).
\]
Taking the arithmetic mean of \eqref{cor key eq new} over all permutations yields the first statement. 
The second statement follows by a similar argument.
\end{proof}

\subsection{Robust and weak robust anonymous rules}\label{Anonymous rules}

To discuss the implications of the difference between $P$-robustness and weak $P$-robustness,
we consider anonymous robust and weak robust rules. %, as well as a difference between robustness and weak robustness. 
%By Proposition \ref{main result}, 
By Proposition \ref{main result}, if $n$ is odd, a SMR is the unique rule that is both robust and anonymous. %\footnote{In Appendix \ref{strong robustness section}, we give another characterization of a SMR with odd $n$ using a stronger version of robustness.}
However, if $n$ is even, no anonymous rule is robust. 
%\red robust \black 
%That is, there is a trade-off between robustness and anonymity. 

\begin{corollary}\label{anonymous robust}
Suppose that $n$ is odd. Then, a deterministic voting rule is robust and anonymous if and only if it is a SMR. 
Suppose that $n$ is even. Then, no deterministic voting rule is both robust and anonymous. 
\end{corollary}
\begin{proof}
A voting rule is robust and anonymous if and only if it is an anonymous WMR with nonnegative weights allowing no ties, which is a SMR with odd $n$. 
\end{proof}

Corollary \ref{anonymous robust} implies that an anonymous rule is not robust if $n$ is even or if it is not a SMR.  
For example, a supermajority rule is not robust because it is anonymous.

%In Appendix \ref{Robustness is too weak a requirement for random rules}, we demonstrate that not only anonymous deterministic rules but also anonymous random rules cannot be $\Delta(\mathcal{X})$-robust when $n$ is even, where we introduce the concept of $\Delta(\mathcal{X})$-robustness for random voting rules.  

Although no anonymous rule is robust when $n$ is even, there exists a weakly robust anonymous rule regardless of $n$ by Proposition \ref{main result}, 
which is an anonymous SMR (a SMR with an anonymous tie-breaking rule). 
\begin{corollary}\label{anonymous weakly robust}
A deterministic voting rule is weakly robust and anonymous if and only if it is an anonymous SMR. 
\end{corollary}
\begin{proof}
A voting rule is weakly robust and anonymous if and only if it is an anonymous WMR with nonnegative weights, which is an anonymous SMR. 
\end{proof}

%Corollaries \ref{anonymous robust} and \ref{anonymous weakly robust} tell us the following: under the trade-off between robustness and anonymity when $n$ is even,  we must be content with a nonanonymous rule if we require robustness and we must be content with a weakly robust rule if we require anonymity. 

Given the above corollaries, we illustrate the difference between robust SMRs and weak robust SMRs as well as the trade-off between robustness and anonymity. Suppose that $n$ is even. 
By Proposition~\ref{main result}, 
a SMR with some tie-breaking rule is robust if and only if it is represented as a WMR allowing no ties, and such a SMR is nonanonymous by Corollary~\ref{anonymous robust}. For example,  a SMR with a casting (tie-breaking) vote is a robust nonanonymous rule. There are two cases to consider. First, assume that the presiding officer with a casting vote is a member of a group of $n$ individuals. This rule is equivalent to a WMR such that the presiding officer's weight is slightly greater than the others' weights.  Next, assume that the presiding officer is not a member of a group of $n$ individuals and that he or she votes only when there is a tie.  This rule is equivalent to 
a WMR with $n+1$ individuals including the presiding officer such that the presiding officer's weight is very close to zero. 
Each of these WMRs does not have ties and is robust

On the other hand, a SMR with any tie-breaking rule is weakly robust by Proposition~\ref{main result}. 
In particular, a SMR with the status quo tie-breaking rule (i.e.,\ the status quo is followed whenever there is a tie)  is a weakly robust anonymous rule, which is not robust when $n$ is even by Corollary~\ref{anonymous robust}. 

In summary, we must be content with a nonanonymous rule 
such as a SMR with a casting vote if we require robustness; we must be content with a weakly robust rule such as a SMR with the status quo tie-breaking rule if we require anonymity. 
A SMR with a casting vote is used by legislatures such as the United States Senate, the Australian House of Representatives, and the National Diet of Japan. A SMR with the status quo tie-breaking rule is used by legislatures such as the New Zealand House of Representatives, the British House of Commons, and the Australian Senate.

\section{Robustness vs.\ efficiency}\label{Robustness vs. efficiency}

Although a $P$-robust rule is not necessarily efficient, any robust rule is efficient because it is a WMR. This section clarifies the relationship between robustness, efficiency, and WMRs.
 
\subsection{Equivalence of robustness and efficiency}\label{equivalence of robustness and efficiency}
A voting rule $\phi \in \Phi$ is said to be strictly efficient under $p \in \Delta(\mathcal{X})$ if no other random voting rule is weakly Pareto-preferred to $\phi$.
The next lemma shows that strict efficiency is equivalent to robustness as a requirement for voting rules.\footnote{This subsection builds on a suggestion by a reviewer.}

\begin{lemma}\label{summary WMRs}
Fix $p^\circ\in \Delta(\mathcal{X})^\circ\equiv \{p\in\Delta(\mathcal{X}): p(x)>0 \text{\em \ for each }x\in\mathcal{X}\}$, where all $x$ is possible.  
The following conditions for $\phi\in \Phi$ are equivalent to each other.

\begin{enumerate}[\em (i)]
\item 
$\phi$ is strictly efficient under $p^\circ$. 
\item $\phi$ is strictly efficient under any $p\in \Delta(\mathcal{X})$.
\item $\phi$ is robust.
%\item $\phi$ is a WMR with nonnegative weights such that there are no ties. 
\end{enumerate}
\end{lemma}

The equivalence of (i) and (ii) is established by \citet{azrielikim2016}.\footnote{See Propositions 1 and 5 in \citet{azrielikim2016}.} That is, strict efficiency in the set of random voting rules under a fixed probability distribution implies strict efficiency under any probability distribution, which is a stronger requirement than robustness. This implies that a strictly efficient rule is robust.  The converse is also true: a robust rule is strict efficient, as shown in Appendix~\ref{alternative proof}. Therefore, strict efficiency is equivalent to robustness.

\subsection{The Rae-Taylor-Fleurbaey theorem}\label{RTF theorem}

This subsection reviews the characterization of WMRs established by \citet{rae1969}, \citet{taylor1969}, and \citet{fleurbaey2008}, and explains how it relates to our main results.

By Lemma \ref{weighted majority characterization}, 
$\phi\in\Phi$ is a WMR with a weight vector $w\in \mathbb{R}^n$ 
if and only if 
$\phi(x)\sum_{i\in N}w_ix_i=|\phi(x)\sum_{i\in N}w_ix_i|$ for all $x\in\mathcal{X}$, which is equivalent to the following inequality: 
 for all $\phi'\in\Phi$ and  $x\in\mathcal{X}$, 
\begin{equation}
\phi(x)\sum_{i\in N}w_ix_i
=\left|\phi(x)\sum_{i\in N}w_ix_i\right|
=\left|\phi'(x)\sum_{i\in N}w_ix_i\right|
\geq \phi'(x)\sum_{i\in N}w_ix_i. 
\label{weighted majority 0}
\end{equation}
This is true if and only if, 
for all $p\in\Delta(\mathcal{X})$, 
\begin{equation}
\sum_{i\in N}w_iE_p[\phi(x)x_i]=\max_{\phi'\in\Phi} \sum_{i\in N}w_iE_p[\phi'(x)x_i],\label{weighted majority 3}
\end{equation}
or equivalently, 
\begin{equation}
\sum_{i\in N}w_ir_i(\phi,p)=\max_{\phi'\in\Phi} \sum_{i\in N}w_ir_i(\phi',p).\label{Fleurbaey eq}
\end{equation}
That is, a necessary and sufficient condition for a voting rule to be a WMR is that it maximizes the corresponding weighted sum of responsiveness over all voting rules for each $p\in \Delta(\mathcal{X})$. 
This result is summarized in the following proposition due to \citet{fleurbaey2008},\footnote{See also \citet{brighousefleurbaey2010}, who discuss the implication of this result for democracy.} where 
the sufficient condition is weaker.\footnote{To see why a weaker condition suffices, 
suppose that $\phi$ is not a WMR. Then, (\ref{weighted majority 0}) does not hold for some $\phi'\in\Phi$ and  $x\in\mathcal{X}$, which contradicts (\ref{weighted majority 3}) and (\ref{Fleurbaey eq}) for each $p\in \Delta(\mathcal{X})^\circ$.}   
\begin{proposition}\label{Fleurbaey}
If $\phi$ is a WMR with a weight vector $w$, then {\em (\ref{Fleurbaey eq})} holds for each $p\in \Delta(\mathcal{X})$. 
For fixed $p^\circ\in \Delta(\mathcal{X})^\circ$,  if {\em (\ref{Fleurbaey eq})} holds, then 
$\phi$ is a WMR with a weight vector $w$.
 \end{proposition}

We call the above result the Rae-Taylor-Fleurbaey (RTF) theorem because it generalizes the Rae-Taylor theorem\footnote{See references in \citet{fleurbaey2008}.} 
which focuses on a SMR. 
The normative implication of the RTF theorem is efficiency of WMRs,\footnote{This issue is not formally discussed in \citet{fleurbaey2008}. Instead, \citet{fleurbaey2008} considers the optimality of a WMR by assuming that $w_i$ is proportional to $i$'s utility, where the weighted sum of responsiveness is the total sum of expected utilities.} as summarized in Appendix \ref{efficiency and RTF theorem}.
In particular, a voting rule is strictly efficient if and only if it is a WMR without ties. 

Our main results can be viewed as a generalization of the equivalence of strict efficiency and WMRs, given the equivalence of robustness and strict efficiency established in Lemma~\ref{summary WMRs}. This is because $P$-robustness is a weaker requirement than strict efficiency when $P\subsetneq\Delta(\mathcal{X})$, and Proposition~\ref{main result} implies that a voting rule is robust if and only if it is a WMR without ties.

\section{Discussion}\label{section: discussions}

In this section, we discuss two extensions of Definition \ref{def: robust}.

\subsection*{Random voting rules}
\label{random rule remark}
We focus on deterministic voting rules in the definition of $P$-robustness,  but the same requirement can also be applied to random voting rules. 
However, this requirement is very weak in the case of random voting rules,  even when $P = \Delta(\mathcal{X})$; that is, under the strongest version of the requirement. 
To demonstrate this, Appendix~\ref{Robustness is too weak a requirement for random rules} characterizes the class of robust random rules and shows that, for any such rule, there exists a deterministic rule that is Pareto-preferred to it under a certain probability distribution of preferences.

\subsection*{Heterogeneous priors}

We adopt the notion of Pareto-dominance under common beliefs in Definition \ref{def: robust}.
This is because Pareto-dominance under heterogeneous beliefs is not normatively compelling as discussed in the introduction. 
In addition, when we adopt Pareto-dominance under heterogeneous beliefs, robustness implies dictatorship. 
Let us say that a voting rule $\phi\in \Phi$ is robust under heterogeneous priors if, for any  $p_1,\ldots,p_n\in \Delta(\mathcal{X})$, it holds that $r_i(\phi,p_i)>1/2$ for at least one $i\in N$. 
We show that there exists a dictator, i.e., individual $i\in N$ with $\phi(x)=x_i$ for all $x\in \mathcal{X}$. 
Seeking a contradiction, suppose that no individual is a dictator. 
Then, for each $i\in N$, there exists $x^i\in \mathcal{X}$ and $p_i\in\Delta(\mathcal{X})$ 
such that $\phi(x^i)\neq x_i^i$ and $p_i(x^i)=1$; that is, $r_i(\phi,p_i)=0\leq 1/2$. % $p_i(\phi(x^i)=x_i^i)=0\leq 1/2$. 
Therefore, $\phi$ is not robust under heterogeneous priors. 
On the other hand, 
if individual $i$ is a dictator, then $r_i(\phi,p)=1>1/2$ for all $p\in \Delta(\mathcal{X})$, 
so $\phi$ is robust under heterogeneous priors.

\section{Conclusion}\label{Conclusion}

The justification of WMRs and, in particular, a SMR based on efficiency arguments or axiomatic characterizations has yielded some of the celebrated contributions to the social choice and voting literature. The two paramount examples rationalizing a SMR within a dichotomous setting are Condorcet's jury theorem and May's theorem,\footnote{See \citet{may1952}, \citet{fishburn1973}, and \citet{dasguptamaskin2008}.} where the rationalization of a voting rule is based on asymptotic (i.e., infinite-individual) probabilistic criteria or deterministic criteria. An alternative approach based on non-asymptotic (i.e., finite-individual) probabilistic criteria was pioneered by \citet{rae1969}, who suggested aggregate responsiveness as a meaningful criterion for evaluating the performance of a voting rule in the constitutional stage, namely, where the veil of ignorance prevails.

This paper contributes to the latter literature by introducing the concept of $P$-robustness. The $P$-robustness criterion requires that a voting rule avoid the worst-case scenario in which the true responsiveness of every individual is less than or equal to one-half under every distribution of preferences in $P$. 
The justification for the $P$-robustness property is based on two arguments: (i) it prevents the collective decision from systematically reflecting only minority preferences on average, and (ii) it rules out Pareto inferiority to the inverse rule. In addition, $P$-robustness bypasses the problem of spurious unanimity, which arises under heterogeneous beliefs.

We establish that a voting rule is $P$-robust if and only if there exists a nonnegative weight vector such that the weighted average of individual responsiveness exceeds one-half under every extreme point of $P$. While every strictly efficient rule is $P$-robust for arbitrary $P$, the converse does not hold: a $P$-robust rule need not be efficient. However, when $P$ is the set of all distributions, $P$-robust rules coincide with WMRs without ties, which are strictly efficient. In this case, $P$-robustness and strict efficiency are equivalent---a fact that can be established directly, without relying on the robustness or efficiency properties of WMRs.

Given the equivalence of robustness and strict efficiency, our main result can be viewed as a generalization of the equivalence of strict efficiency and WMRs. We introduce a weaker criterion (i.e., $P$-robustness) and characterize the set of voting rules that satisfy it. This characterization includes, as a special case, the equivalence between strict efficiency (i.e., robustness) and WMRs. In this sense, our analysis extends the efficiency-based justification of WMRs, as formalized in the RTF theorem.

\bigskip
\appendix
\setcounter{section}{0}
\setcounter{theorem}{0}
\setcounter{lemma}{0}
\setcounter{claim}{0}
\setcounter{proposition}{0}
\setcounter{definition}{0}
\renewcommand{\thetheorem}{\Alph{theorem}}
\renewcommand{\thelemma}{\Alph{lemma}}
\renewcommand{\theclaim}{\Alph{claim}}
\renewcommand{\theproposition}{\Alph{proposition}}
\renewcommand{\thedefinition}{\Alph{definition}}

\begin{center}
\Large{{\bf Appendix}}
\end{center}

\section{Robust random voting}\label{Robustness is too weak a requirement for random rules}

In this appendix, we define and characterize a $P$-robust random voting rule with $P=\Delta(\mathcal{X})$, i.e., a robust random rule.  Robustness in this case is a very weak requirement. However, there is no anonymous random voting rule that is robust when the number of individuals is even.

Recall that a random voting rule is a mapping $\bar\phi:\mathcal{X}\to [-1,1]$ assigning a collective decision $1$ with probability $(1+\bar\phi(x))/2$ and $-1$ with probability $(1-\bar\phi(x))/2$ to each $x\in \mathcal{X}$. 
Thus, the conditional responsiveness of individual $i\in N$ given $x\in \mathcal{X}$ is ${(1+\bar\phi(x))}/{2}$ if $x_i=1$ and ${(1-\bar\phi(x))}/{2}$ if $x_i=-1$, which is rewritten as ${(\bar\phi(x)x_i+1)}/{2}$. 
Therefore, when $x\in \mathcal{X}$ is drawn according to $p\in \Delta(\mathcal{X})$, the responsiveness of individual $i\in N$ is calculated as $(E_p[\bar\phi(x)x_i]+1)/2$, which is analogous to (\ref{fundamental 2}).

We define a robust random rule by applying Definition \ref{def: robust} and provide a characterization.

\begin{definition}\label{random robust}
A random voting rule $\bar\phi\in\Delta(\Phi)$ is {\em robust} if, for each  $p\in\Delta(\mathcal{X})$, the responsiveness of at least one individual is strictly greater than one-half, 
%$(E_p[\bar\phi(x)x_i]+1)/2> 1/2$ for at least one $i\in N$, 
or equivalently, $E_p[\bar\phi(x)x_i]> 0$ for at least one $i\in N$. 
%the responsiveness of at least one individual is strictly greater than one-half. 
	\end{definition}

\begin{proposition}\label{main random result}
A random voting rule $\bar\phi\in \Delta(\Phi)$ 
 is robust if and only if there exists $w\in \mathbb{R}_+^n$ satisfying 
\begin{equation}
\bar\phi(x)\gtrless 0 \ \Leftrightarrow \ \sum_{i\in N}w_ix_i\gtrless 0.	 \label{random WMR}
\end{equation}
\end{proposition}
\begin{proof}
The proof is essentially the same as that of Proposition \ref{main result}. 
We enumerate elements in $\mathcal{X}$ as $\{x^j\}_{j\in M}$, where $M\equiv\{1,\ldots,m\}$ is an index set with $m\equiv 2^n$. 
Consider an $n\times m$ matrix 
$\bar L=[\bar l_{ij}]_{n\times m}=\left[\bar\phi(x^j)x_i^j\right]_{n\times m}$. 
Then, (\ref{random WMR}) is rewritten as 
$\sum_{i\in N} w_i \bar l_{ij}=\sum_{i\in N}w_i\big(\bar\phi(x^j)x_i^j\big)>0$ for each $j\in M$, 
or equivalently, 
$w^\top \bar L\gg 0$.
By Lemma \ref{alternative1}, 
this is true if and only if there does not exist $p=(p_j)_{j\in M}> 0$  such that $\bar Lp\leq 0$, or equivalently, $
\sum_{j\in M} \bar l_{ij}p_j
=\sum_{j\in M}\bar\phi(x^j)x_i^j
p_j\leq 0$ for each $i\in N$; that is, $\bar \phi$ is robust. 
\end{proof}

Not only a WMR 
but also many other random voting rules satisfy  (\ref{random WMR}). 
For example, assume that $n$ is odd and 
let $\bar\phi\in \Delta(\Phi)$ be a random voting rule such 
that the collective decision is the majority's vote with probability $0.5+\varepsilon$, where $0<\varepsilon<1/2$; that is, 
\begin{equation}
	\bar \phi(x)=
	\begin{cases}
		+2\varepsilon &\text{ if } \sum_{i\in N}x_i> 0,\\
		-2\varepsilon &\text{ if } \sum_{i\in N}x_i< 0.
	\end{cases} \notag
\end{equation}
Note that $\bar\phi$ satisfies 
(\ref{random WMR}) with $w_i=1$ for all $i\in N$, and thus it is robust.
However, $\bar\phi$ can be dominated by a deterministic voting rule. 
To see this, let $\phi\in \Phi$ be a SMR, and let $p\in \Delta(\mathcal{X})$ be such that $p(x)=1$ if $x_i=1$ for all $i\in N$ and $p(x)=0$ otherwise. 
Then, for each $i\in N$, $E_p[\bar\phi(x)x_i]-E_p[\phi(x)x_i]=2\varepsilon-1<0$,
which implies that $\phi$ is strictly Pareto-preferred to $\bar\phi$ under $p$ in terms of responsiveness.

Although robustness is a weak requirement for random voting rules, there exists no anonymous robust random rule when $n$ is even. 
Let $p\in \Delta(\mathcal{X})$ be such that 
\[
p(x)
=
\begin{cases}
1/\binom{n}{n/2}	& \text{ if }\#\{i:x_i=1\}=n/2,\\
0 & \text{ otherwise}.
\end{cases}
\]
Then, for any anonymous random voting rule $\bar\phi\in \Delta(\Phi)$, it holds that 
$E_p[\bar\phi(x)x_i]=E_p[\bar\phi(x)x_j]$ for all $i,j\in N$ and 
\[
\sum_{i\in N}E_p[\bar\phi(x)x_i]
=E_p\left[\bar\phi(x)\sum_{i\in N} x_i\right]=0.
\]
This implies that $(E_p[\bar\phi(x)x_i]+1)/2=1/2$ for all $i\in N$; that is, the responsiveness of every individual equals one-half. 
Therefore, $\bar\phi$ is not robust.

\section{$P$-robustness with an arbitrary convex set $P$}
\label{an arbitrary convex set}

We extend Proposition \ref{main result} to the case where $P$ is an arbitrary closed convex set as its corollary.

\begin{corollary}\label{main result cor}
Let $P\subseteq \Delta(\mathcal{X})$ be a closed convex set with the set of extreme points $\ext(P)$. 
Then, the following results hold. 
\begin{enumerate}[\em (1)]
\item A voting rule $\phi\in\Phi$ is $P$-robust if and only if there exists $w\in \mathbb{R}_{+}^n$ such that 
\begin{equation}
\frac{\sum w_i r_i(\phi,p)}{\sum w_i}> 1/2\ \text{ for all }\ p\in \ext(P).
\label{robustness condition w1'}    	
\end{equation}

\item  A voting rule $\phi\in \Phi$ is weakly $P$-robust if and only if there exists $w\in \mathbb{R}_{+}^n$ such that 
\begin{equation}
\frac{\sum w_i r_i(\phi,p)}{\sum w_i}\geq 1/2 \ \text{ for all }\ p\in \ext(P).
\label{robustness condition w2'}    	
\end{equation}
\end{enumerate}
\end{corollary}

\begin{proof}

\begin{comment}
Note that \eqref{robustness condition w1'} and \eqref{robustness condition w2'}  are rewritten as 
\begin{align}
\sum_{i\in N}w_iE_{p}[\phi(x) x_i]  > 0\text{ for all }p\in \ext(P),
\label{robustness condition w1''}  \\  	
\sum_{i\in N}w_iE_{p}[\phi(x) x_i]  \geq 0\text{ for all }p\in \ext(P),
\label{robustness condition w2''}    	
\end{align}
respectively. 
\end{comment}

The ``if'' part of each result is immediate from the definitions of $P$-robustness and weak $P$-robustness. 
Thus, it is enough to prove the ``only if'' part. 

To prove the ``only if'' part of the first result, suppose that $\phi$ is $P$-robust. 
Then,  by definition, 
$\min_{p\in P}\max_{i\in N}r_i(\phi,p)
> 1/2$. 
It is known in convex geometry that there exists a convex polytope $P'\subset  \Delta(\mathcal{X})$ such that $P\subset P'$ and their Hausdorff distance is arbitrarily small.\footnote{See, for example, \cite{bronshteynivanov1975}.}
In particular, we can choose $P'$ that satisfies $\min_{p\in P'}\max_{i\in N}r_i(\phi,p)> 1/2$. 
Since $P'$ is the convex hull of a finite set, Proposition \ref{main result} implies that there exists $w\in \mathbb{R}_{+}^n$ satisfying \eqref{robustness condition w1} for each extreme point of $P'$, which also satisfies \eqref{robustness condition w1'} since $\ext(P)\subset P'$.    	

To prove the ``only if'' part of the second result, suppose that there is no $w\in \mathbb{R}_+^n$ satisfying \eqref{robustness condition w2'}.  
Then, $\max_{w\in \Delta(N)}\min_{p\in P}\sum_{i\in N}w_ir_i(\phi,p) < 1/2$.
It is also known in convex geometry that there exists a convex polytope $P'\subset  \Delta(\mathcal{X})$ such that $P\supset P'$ and their Hausdorff distance is arbitrarily small. 
In particular, we can choose $P'$ that satisfies 
$\max_{w\in \Delta(N)}\min_{p\in P'}\sum_{i\in N}w_ir_i(\phi,p) < 1/2$.
Since $P'$  is the convex hull of a finite set, $\phi$ is not weakly $P'$-robust by Proposition \ref{main result}, which in turn implies that $\phi$ is not weakly $P$-robust since $P'\subset P$. 
\end{proof}

\section{Proof of Lemma \ref{summary WMRs}}\label{alternative proof}

The equivalence of (i) and (ii) is implied by Propositions 1 and 5 in \citet{azrielikim2016}, but for completeness, we provide a full proof below, following a suggestion by a reviewer.

We begin by showing that (i) implies (ii).
Fix $p \in \Delta(\mathcal{X})^{\circ}$ and suppose that $\phi$ is not strictly efficient under $q\in \Delta (\mathcal{X})$.
%Then, by definition, there exists $q \in \Delta(\mathcal{X})$ such that the inverse rule $-\phi$ weakly Pareto-preferred to $\phi$.
Then, by definition, there exists a random voting rule $\bar\phi \in \Delta(\Phi)$ such that $\bar\phi$ is weakly Pareto-preferred to $\phi$ under $q$.
For any $x \in \mathcal{X}$, let $r(x)=q(x)/p(x)$, which is well-defined by $p \in \Delta(\mathcal{X})^{\circ}$.
Define a random voting rule $\bar \phi' \in \Delta(\Phi)$ such that $\bar \phi'(x) \equiv \bigl( r(x)/R \bigr) \bar\phi(x)+\bigl(1- r(x)/R \bigr)\phi(x)$ for any $x \in \mathcal{X}$, where $R=\max_{x \in \mathcal{X}} r (x)>0$. 
Then, we can see that
\begin{align*}
E_p[\bar \phi'(x)x_i]-E_p[\phi(x)x_i]&= E_p\bigl[ \bigl( r(x)/R \bigr) \bar\phi(x)x_i+\bigl(1- r(x)/R \bigr)\phi(x)x_i\bigr]-E_p[\phi(x)x_i \bigr]  \\       
&= \frac{1}{R} \Bigl( E_q[ \bar\phi(x)x_i]-E_q[ \phi(x)x_i] \Bigr)\geq 0
\end{align*}
for any $i \in N$, which implies that $\bar \phi'$ is weakly Pareto-preferred to $\phi$ under $p$ because $\bar\phi$ is weakly Pareto-preferred to $\phi$ under $q$.

Next, observe that (ii) implies (iii) by definition.

Finally, we show that (iii) implies (i).
Suppose that there exists a random voting rule $\bar \phi \in \Delta(\Phi)$ and $p \in \Delta(\mathcal{X})$ such that $\bar \phi$ is weakly Pareto-preferred to $\phi$ under $p \in \Delta(\mathcal{X})$. 
Define $q \in \Delta(\mathcal{X})$ by
\[
q(x)=
\frac{p(x)\alpha(x)}{Q},
\]
where $\alpha(x)=(1-\bar\phi(x)/\phi(x))/2\in [0,1]$ for all $x\in \mathcal{X}$ and 
$Q=\sum_{x}p(x)\alpha(x)$.
Then, 
\begin{align*}
E_q[(-\phi)(x)x_i]-E_q[\phi(x)x_i]
&=  \sum_{x}\frac{p(x)\alpha(x)}{Q} \bigl((-\phi)(x)x_i-\phi(x)x_i \bigr)\\
&=  \sum_{x}\frac{p(x)}{Q} \Bigl( (1-2\alpha(x))\phi(x)x_i-\phi(x)x_i \Bigr)\\
&=  \sum_{x}\frac{p(x)}{Q} \Bigl( \bar\phi(x)x_i-\phi(x)x_i \Bigr)\\
&= \bigl( E_p[\bar \phi(x)x_i]-E_p[\phi(x)x_i] \bigr)/Q \ge  0
\end{align*}
for any $i \in N$, which implies that the inverse rule $-\phi$ is weakly Pareto-preferred to $\phi$ under $q$, and hence, $\phi$ is not robust.

\section{The normative implication of the RTF theorem}\label{efficiency and RTF theorem}

To discuss the normative implication of RTF theorem, 
we consider the following notions of efficiency, including strict efficiency discussed in Section \ref{Robustness vs. efficiency}.

\begin{definition}
Fix $p\in \Delta(\mathcal{X})$. 
A deterministic voting rule $\phi\in \Phi$ is efficient  under $p$ if any random voting rule is not Pareto-preferred to $\phi$ in terms of responsiveness under $p$. 
A deterministic voting rule $\phi\in \Phi$ is weakly efficient  under $p$ if any random voting rule is not strictly Pareto-preferred to $\phi$ in terms of responsiveness under $p$. 
\end{definition}

Using the RTF theorem, we can obtain the following characterizations of WMRs. 

%\footnote{Another normative implication of Proposition \ref{Fleurbaey} is optimality of WMRs in terms of Paretian social preferences, which is immediate from Harsanyi's utilitarianism theorem \citep{harsanyi1955}. See Appendix \ref{Optimality of WMRs}.}
\begin{proposition}\label{Fleurbaey 2}
Fix $p^\circ\in \Delta(\mathcal{X})^\circ$. 
A deterministic voting rule is strictly efficient under $p^\circ$ if and only if it is a WMR with nonnegative weights such that there are no ties. 
A deterministic voting rule is efficient under $p^\circ$ if and only if it is a WMR with positive weights. 
A deterministic voting rule is weakly efficient under $p^\circ$ if and only if it is a WMR with nonnegative weights. 
\end{proposition}

%\textcolor{red}{The first statement of Proposition \ref{Fleurbaey 2} implies the equivalence of (i) and (ii) in Proposition \ref{summary WMRs}. }
Note that a WMR with nonnegative weights allowing no ties can be represented as a WMR with positive weights allowing no ties. This is because the set of all weight vectors representing a WMR $\phi$ allowing no ties is $\{w\in \mathbb{R}^n:\sum_{i\in N} w_i(\phi(x)x_i)>0\text{ for all }x\in \mathcal{X}\}$ by Lemma \ref{weighted majority characterization}, which is an open set.

To the best of the authors' knowledge, the formal proof has never appeared in the literature, so we give the proof based on the key equations of the RTF theorem, (\ref{weighted majority 0}) and (\ref{weighted majority 3}).

\begin{proof}
We first prove the first statement, which asserts that exactly one of the following holds. 
\begin{enumerate}[(a)]
\item 
A voting rule $\phi$ is a WMR with nonnegative weights allowing no ties. By Lemma \ref{weighted majority characterization}, this is true if and only if there exists $w=(w_i)_{i\in N}\geq 0$ such that, for all $\phi'\in\Phi\setminus\{\phi\}$ and  $x\in\mathcal{X}$, (\ref{weighted majority 0}) is true with strict inequality holding for at least one decision profile $x$. This is true if and only if there exists $w=(w_i)_{i\in N}\geq 0$ such that
\begin{equation*}
\sum_{i\in N}w_iE_{p^\circ}[\phi(x)x_i]>\sum_{i\in N}w_iE_{p^\circ}[\phi'(x)x_i] \text{ for all $\phi'\in\Phi\setminus\{ \phi\}$},  
\end{equation*}
or equivalently, 
\begin{equation*}
\sum_{i\in N}w_i(E_{p^\circ}[\phi(x)x_i]-E_{p^\circ}[\phi'(x)x_i])
> 0 \text{ for all $\phi'\in\Phi\setminus\{ \phi\}$}.
%\text{ for all $x\in\mathcal{X}$}.
\end{equation*}

\item 
There exists a random voting rule that is weakly Pareto-preferred to $\phi$ if and only if there exists $\rho\in \mathbb{R}_+^{\Phi\setminus\{\phi\}}$ such that 
\[
E_{p^\circ}[\phi(x)x_i]\leq 
\sum_{\phi'\in \Phi\setminus\{\phi\}}E_{p^\circ}[\phi'(x)x_i]\rho(\phi')/\sum_{\phi'\in \Phi\setminus\{\phi\}}\rho(\phi')\text{ for all $i\in N$}, 
\]
or equivalently, 
\[
\sum_{\phi'\in \Phi\setminus\{\phi\}}(E_{p^\circ}[\phi(x)x_i]-E_{p^\circ}[\phi'(x)x_i])\rho(\phi')\leq 0 
\text{ for all $i\in N$}.
\]
\end{enumerate}
By Lemma \ref{alternative1}, exactly one of 
(a) and (b) holds. 
%\textcolor{red}{Although the same theorem of alternatives characterizes both $P$-robust rules and strictly efficient rules, its use is different from each other. In fact, as discussed in Section \ref{Robustness vs. efficiency}, robustness and strict efficiency are different conditions.}

The second and third statements are implied by the well-known theorem of \citet{wald1950} on admissible decision functions (or that of \citet{pearce1984} on undominated strategies). Proposition~\ref{Fleurbaey} states that a voting rule is a WMR with a weight vector $w$ if and only if (\ref{weighted majority 3}) holds. Mathematically, (\ref{weighted majority 3}) is equivalent to expected utility maximization, where $\Phi$ is the set of actions, $N$ is the set of states, and $w_i/\sum_j w_j$ is a probability of state $i\in N$.  Therefore, we can apply the theorem of \citet{wald1950} on admissible decision functions. In particular, Theorems 5.2.1 and 5.2.5 in \citet{blackwellgirshick1954} are useful. Theorem 5.2.1 implies that a voting rule is weakly efficient if and only if there exists a weight vector $w>0$ such that (\ref{weighted majority 3}) holds. Theorem 5.2.5 implies that a voting rule is efficient if and only if there exists a weight vector $w\gg 0$ such that (\ref{weighted majority 3}) holds. Therefore, this proposition holds by Proposition~\ref{Fleurbaey}.
\end{proof}

\section*{Declarations}

\subsection*{Funding}
Ui and Nakada were supported by MEXT, Grant-in-Aid for Scientific Research.

\subsection*{Conflict of interest}
The authors declare that they have no conflict of interest.

\end{document}

\title{Online Appendix to\\
 ``Robust Voting under Uncertainty''}

\author{Satoshi Nakada\thanks{School of Management, Department of Business Economics, Tokyo University of Science.}\and  Shmuel Nitzan\thanks{Department of Economics, Bar-Ilan University.}\and Takashi Ui\thanks{Hitotsubashi Institute for Advanced Study, Hitotsubashi University; oui@econ.hit-u.ac.jp.}}

\date{July 2025}
\maketitle

\setcounter{section}{0}
\setcounter{theorem}{0}
\setcounter{lemma}{0}
\setcounter{claim}{0}
\setcounter{proposition}{0}
\setcounter{definition}{0}
\renewcommand{\thetheorem}{\Alph{theorem}}
\renewcommand{\thelemma}{\Alph{lemma}}
\renewcommand{\theclaim}{\Alph{claim}}
\renewcommand{\theproposition}{\Alph{proposition}}
\renewcommand{\thedefinition}{\Alph{definition}}

\renewcommand{\theequation}{\thesection.\arabic{equation}}
\newcommand{\red}{\color{red}}
\newcommand{\black}{\color{black}}
\newcommand{\blue}{\color{blue}}

In the main text, we introduce the concept of $P$-robustness by assuming that an individual votes sincerely and prefers a voting rule with higher responsiveness.
We provide justification for these assumptions by discussing two alternative definitions of $P$-robustness in this online appendix, mainly focusing on the case of $P=\Delta(\mathcal{X})$. 
First, we redefine robustness as a requirement for a voting mechanism, where individuals can vote insincerely, and show that sincere voting is weakly dominant for every individual if the mechanism satisfies robustness. %; that is, robustness implies strategyproofness. 
Next, we redefine robustness using the expected utility of individuals (rather than responsiveness) and provide justification for the original definition, where we consider a set of probability distributions over individuals' utility functions.

\appendix

\section{Robustness as a requirement for a mechanism}\label{Robustness as a requirement for a mechanism}
We regard $\phi\in \Phi$ as a  mechanism with a message space $\mathcal{X}$.
Individual $i\in N$ has a VNM utility function $u_i:\{-1,1\}\to \mathbb{R}$ over the set of alternatives $\{-1,1\}$ representing a strict preference relation, i.e., $u_i(1)\neq u_i(-1)$. 
The preferred alternative of individual $i$ is denoted by $\delta u_i\in \{-1,1\}$, i.e.,   
\[
\delta u_i\equiv 
\begin{cases}
1 & \text{ if } u_i(1)>u_i(-1),\\
-1 & \text{ if } u_i(-1)>u_i(1). 
\end{cases}
\]
%and the decision profile is $\delta u=(\delta u_i)_{i\in N}$. 
Let $\mathcal{U}_i\equiv \{u_i: u_i(1)\neq u_i(-1)\}$ be the set of all utility functions of individual $i$ representing strict preferences. 
We write $\mathcal{U}\equiv \prod_{i\in N}\mathcal{U}_i$. 

%Let $\mathcal{U}_i$ and $\mathcal{U}\equiv \prod_{i\in N}\mathcal{U}_i$ denote the set of all utility functions of individual $i$ with strict preferences and the set of all such utility function profiles, respectively. 

Individual $i\in N$ with a voting strategy $\sigma_i:\mathcal{U}_i\to \{-1,1\}$ casts a vote for $\sigma_i(u_i)\in \{-1,1\}$. 
A sincere voting strategy is $\delta_i:\mathcal{U}_i\to\{-1,1\}$ satisfying $\delta_i(u_i)=\delta u_i$ for all $u_i\in \mathcal{U}_i$. 
Let $\Sigma_i$ be the set of all strategies of individual $i$. 
We write $\Sigma=\prod_{i\in N}\Sigma_i$.

Assume that $u\equiv (u_i)_{i\in N}\in \mathcal{U}$ is randomly drawn according to a probability distribution $\lambda\in \Delta(\mathcal{U})$, where $\Delta(\mathcal{U})$ is the set of all probability distributions over $\mathcal{U}$.  
When the individuals follow a strategy profile $\sigma\equiv (\sigma_i)_{i\in N}\in \Sigma$ under a mechanism $\phi$, the probability distribution of a decision profile $\sigma(u)\equiv (\sigma_i(u_i))_{i\in N}$ is $\lambda\circ \sigma^{-1}\in \Delta(\mathcal{X})$ with 
\[
\lambda\circ \sigma^{-1}(x)\equiv 
\lambda(\{u\in \mathcal{U}:\sigma(u)=x\}),  
\]
and the responsiveness of individual $i\in N$ is 
\begin{equation}
\lambda\circ\sigma^{-1}(\phi(x)=x_i)
\equiv 
\lambda\circ\sigma^{-1}(\{x\in \mathcal{X}:\phi(x)=x_i\})
=\lambda(\{u\in \mathcal{U}:\phi(\sigma(u))=\sigma_i(u_i)\}). \label{reported responsiveness}	
\end{equation}
%which is based upon the reported preferences. 

We define robustness of a voting mechanism using the responsiveness (\ref{reported responsiveness}) generated by the reported messages $\sigma(u)$, where the true probability distribution $\lambda \in \Delta(\mathcal{U})$ and the  strategy profile $\sigma\in \Sigma$ followed by the individuals are assumed to be unknown. 

\begin{definition}\label{def: robust mechanism}
A voting rule $\phi\in\Phi$ is {\em robust as a mechanism} if, for each $\lambda\in\Delta(\mathcal{U})$ and each strategy profile $\sigma\in \Sigma$, $\lambda\circ\sigma^{-1}(\phi(x)=x_i)> 1/2$ for at least one $i\in N$. 
%A voting rule $\phi\in\Phi$ is {\em robust as a mechanism under sincere voting} if, for each $\lambda\in\Delta(\mathcal{U})$, $\lambda\circ\delta^{-1}(\phi(x)=x_i)=\lambda(\{u\in \mathcal{U}:\phi(\delta u)=\delta u_i\})> 1/2$ for at least one $i\in N$, where $\delta u=(\delta u_j)_{j\in N}$. 
\end{definition}

If $\phi\in\Phi$ is robust as a mechanism, then it is robust in the original sense. 
This is because $\phi$ is robust if and only if, for each $\lambda\in\Delta(\mathcal{U})$, the responsiveness (\ref{reported responsiveness}) is strictly greater than one-half for at least one individual when every player adopts a sincere voting strategy. 
However, it is straightforward to see that both concepts are equivalent.

\begin{lemma}\label{lemma: robust mechanism}
A voting rule $\phi\in\Phi$ is robust as a mechanism if and only if it is robust. 
\end{lemma}
\begin{proof}
This lemma follows from $
\Delta(\mathcal{X})=\bigcup_{\sigma\in \Sigma}\{\lambda\circ\sigma^{-1}\in \Delta(\mathcal{X}):\lambda\in \Delta(\mathcal{U})\}$.
\end{proof}

%Lemma \ref{lemma: robust mechanism} implies that the assumption of sincere voting is not essential in Definition \ref{def: robust}. 

Moreover, if $\phi$ is robust as a mechanism, then a sincere voting strategy is weakly dominant under $\phi$; that is, $\phi$ satisfies strategy-proofness. 

\begin{lemma}\label{sincere voting lemma}
If $\phi\in \Phi$ is robust as a mechanism, then, for each individual $i\in N$, the sincere voting strategy $\delta_i$ is a best response to any strategy profile of the other individuals.  
\end{lemma}

\begin{proof}
%In the main result, we show that $\phi$ is robust if and only if it is a WMR with nonnegative weights allowing no ties, where a sincere voting strategy is weakly dominant. However, this lemma without using the main result. 
It is enough to show that, 
for any decision profile of the other individuals $x_{-i}\equiv (x_j)_{j\neq i}$, it holds that 
$u_i(\phi(\delta u_i,x_{-i}))\geq u_i(\phi(-\delta u_i,x_{-i}))$ for all $i\in N$ and $u_i\in \mathcal{U}_i$. 
Seeking a contradiction, suppose otherwise. 
That is, there exist $i^*\in N$, $u_{i^*} \in \mathcal{U}_{i^*}$, and $z_{-i^*}\equiv (z_j)_{j\neq i^*}$ such that $\phi(\delta u_{i^*},z_{-i^*})=-\delta u_{i^*}$ and $\phi(-\delta u_{i^*},z_{-i^*})=\delta u_{i^*}$.
Without loss of generality, we assume that $\delta u_{i^*}=1$.
Let $p\in \Delta(\mathcal{X})$ be such that 
$p(1,z_{-i^*})=p(-1,z_{-i^*})=1/2$. 
Then, 
$r_{i^*}(\phi,p)=0$ and 
$r_i(\phi,p)=1/2$ if $i\neq i^*$. 
This implies that $\phi$ is not robust as a mechanism by Lemma \ref{lemma: robust mechanism}.
\end{proof}

In summary, we can give an equivalent definition of robustness without the assumption of sincere voting and show that a sincere voting strategy is weakly dominant under a robust rule, thus providing justification for the assumption of sincere voting in the main text.

\section{Robustness in terms of  expected utility}\label{Robustness in terms of  expected utility}

Under the assumption of sincere voting justified in Section \ref{Robustness as a requirement for a mechanism},  we redefine robustness using the expected utility of individuals and provide justification for the use of responsiveness in the original definition of robustness.

Individual $i\in N$ with a utility function $u_i\in \mathcal{U}_i$ votes for the preferred alternative $\delta u_i\in \{-1,1\}$, so the resulting decision profile is $\delta u=(\delta u_i)_{i\in N} \in \mathcal{X}$. 
When $\lambda\in \Delta(\mathcal{U})$ is the true probability distribution, individual $i\in N$ is assumed to prefer a voting rule $\phi$ to another voting rule $\phi'$ if and only if $E_\lambda[u_i(\phi(\delta u))]\geq E_\lambda[u_i(\phi'(\delta u))]$, where $E_\lambda$ is the expectation operator with respect to $\lambda$. 
Thus, we consider the following Pareto order in terms of expected utility under $\lambda$: 
we say that $\phi$ is weakly Pareto-preferred to $\phi'$ (in terms of expected utility)  under $\lambda$ if $E_\lambda[u_i(\phi(\delta u))]\geq E_\lambda[u_i(\phi'(\delta u))]$ for all $i\in N$. 
A special case is the Pareto order in terms of responsiveness because the responsiveness equals  the expected utility if the utility level is either zero or one. 
To explain it more formally, denote the set of such utility functions and that of the corresponding profiles of utility functions  
by 
\[
\mathcal{U}_i^*=\{u_i \in \mathcal{U}_i: u_i(1),u_i(-1) \in \{0,1\}\}\text{ and } \mathcal{U}^*\equiv \prod_{i\in N}\mathcal{U}_i^*,
\] 
respectively. 
Then, for $p\in \Delta(\mathcal{X})$ and $\lambda_p\in \Delta(\mathcal{U})$ such that $\lambda_p(u)=p(\delta u)$ if $u\in \mathcal{U}^*$ and $\lambda_p(u)=0$ otherwise, 
$\phi$ is Pareto-preferred to $\phi'$ in terms of expected utility under $\lambda_p$ if and only if $\phi$ is Pareto-preferred to $\phi'$ in terms of responsiveness under $p$  because $E_{\lambda_p}[u_i(\phi(\delta u))]=r_i(\phi,p)$.

We consider an extension of robustness using the Pareto order in terms of expected utility, where the true probability distribution $\lambda \in \Delta(\mathcal{U})$ is assumed to be unknown but known to be contained in $\Lambda\subseteq \Delta(\mathcal{U})$.

\begin{definition}\label{def: lambda robust}
For $\Lambda\subseteq \Delta(\mathcal{U})$, a voting rule $\phi\in \Phi$ is {\em $\Lambda$-robust} if the inverse rule $-\phi$ is not weakly Pareto-preferred to $\phi$ in terms of expected utility under any $\lambda \in \Lambda$. 
\end{definition}

%Even if $\phi$ is $\Lambda$-robust, there may exist $\lambda$ and a voting rule $\phi'$ that is Pareto-preferred to $\phi$ under $\lambda$. 
%Thus, the requirement of $\Lambda$-robustness is considered to be a minimum requirement, which is similar to the case of robustness. 

%We compare $\phi$ and the inverse rule $-\phi$ as a minimum requirement, which is in the same spirit as the definition of robustness. 

By the above discussion, $\phi\in \Phi$ is robust if and only if it is $\Lambda^*$-robust with 
\begin{equation*}
\Lambda^*=\{\lambda\in \Delta(\mathcal{U}): \lambda(u)>0 \text{ implies }u\in  \mathcal{U}^*\}
=\{\lambda_p\in \Delta(\mathcal{U}): p\in \Delta(\mathcal{X})\}.
\end{equation*}
In addition, for each $\Lambda\subset \Lambda^*$, 
there exists $P\subseteq\Delta(\mathcal{X})$ such that $\Lambda=\{\lambda_p\in \Delta(\mathcal{U}):p\in P\}$, and thus $\Lambda$-robustness is equivalent to $P$-robustness.

The underlying assumption in $\Lambda$-robustness with $\Lambda\subseteq \Lambda^*$ is that the utility levels are known to be either zero or one. 
However, even if  $\Lambda\supsetneq \Lambda^*$ and the utility levels are unknown under $\Lambda$, $\Lambda$-robustness can be equivalent to robustness.

To demonstrate it, consider 
\[
\Gamma^*=\{\lambda\in \Delta(\mathcal{U}):\frac{E_\lambda[u_i(\delta u_i)-u_i(-\delta u_i)|\delta u=x]}{E_\lambda[u_i(\delta u_i)-u_i(-\delta u_i)|\delta u=x']}= 1\text{ for all }x,x'\in \mathcal{X}\text{ and }i\in N\}
\]
and assume that the true probability distribution is unknown but known to be contained in $\Gamma^*$. 
Note that $E_\lambda[u_i(\delta u_i)-u_i(-\delta u_i)|\delta u=x]$ is the conditional expected net gain of individual $i\in N$ induced by the switch from an unfavorable alternative $-\delta u_i$ 
to a favorable alternative $\delta u_i$ 
given $\delta u=x$. 
Thus, $\Gamma^*$ is the set of all probability distributions over $\mathcal{U}$ such that the conditional expected net gain is the same for all $x\in \mathcal{X}$, and in particular, it holds that $\Lambda^*\subset \Gamma^*$. 
Note that the conditional expected net gain is unknown because it depends upon the unknown probability distribution $\lambda$, which also implies that the utility levels are unknown. 
However, $\Gamma^*$-robustness is equivalent to robustness. 

\begin{lemma}\label{lemma gamma}
A voting rule $\phi\in\Phi$ is $\Gamma^*$-robust if and only if it is robust.	
\end{lemma}
\begin{proof}
For each $\lambda\in\Delta(\mathcal{U})$, let $p_\lambda\in \Delta(\mathcal{X})$ be the probability distribution of the  decision profile $\delta u$; that is, $p_\lambda(x)=\lambda(\{u\in \mathcal{U}:\delta u=x\})$ for all $x\in \mathcal{X}$. 
Then, 
the expected utility of individual $i\in N$ under $\phi\in \Phi$ is 
\begin{align*}
E_\lambda[u_i(\phi(\delta u))]
%& =\sum_{\phi(x)=x_i}p_\lambda(x)E_\lambda[u_i(\phi(\delta u))|\delta u=x]+\sum_{\phi(x)= -x_i}p_\lambda(x)E_\lambda[u_i(\phi(\delta u))|\delta u=x]\\
& =\sum_{x:\phi(x)=x_i}E_\lambda[u_i(\delta u_i)|\delta u=x]p_\lambda(x)
+\sum_{x:\phi(x)=- x_i}E_\lambda[u_i(-\delta u_i)|\delta u=x]p_\lambda(x)\\
&=\sum_{x:\phi(x)=x_i}E_\lambda[u_i(\delta u_i)-u_i(-\delta u_i)|\delta u=x]p_\lambda(x)
+\sum_{x\in \mathcal{X}}E_\lambda[u_i(-\delta u_i)|\delta u=x]p_\lambda(x)\\
& =c_i 
r_i(\phi,p_\lambda)+
%p_\lambda(\phi(x)=x_i) +
E_\lambda[u_i(-\delta u_i)],
\end{align*}
where $c_i=E_\lambda[u_i(\delta u_i)-u_i(-\delta u_i)|\delta u=x]>0$,  which is independent of $x\in \mathcal{X}$. 
Thus, $E_\lambda[u_i(\phi(\delta u))]\geq E_\lambda[u_i(\phi'(\delta u))]$ if and only if 
$r_i(\phi,p_\lambda)\geq r_i(\phi',p_\lambda)$. 
%$p_\lambda(\phi(x)=x_i)\geq p_\lambda(\phi'(x)=x_i)$. 
This implies that 
$\Gamma^*$-robustness is equivalent to robustness. 
\end{proof}

Note that $\Lambda$-robustness with $\Lambda \supset \Gamma^*$ is a stronger requirement than robustness. 
However, if $\Lambda$ is sufficiently close to $\Gamma^*$, $\Lambda$-robustness remains equivalent to robustness. 
To see this, for $\varepsilon> 0$, let 
\[
\Gamma(\varepsilon)=\{\lambda\in \Delta(\mathcal{U}):\frac{E_\lambda[u_i(\delta u_i)-u_i(-\delta u_i)|\delta u=x]}{E_\lambda[u_i(\delta u_i)-u_i(-\delta u_i)|\delta u=x']}< 1+\varepsilon\text{ for all }x,x'\in \mathcal{X}\text{ and }i\in N\}.
\]
Note that $\Gamma^*\subsetneq \Gamma(\varepsilon)$ and $\Gamma^*=\bigcap_{\varepsilon>0}\Gamma(\varepsilon)$. 
Because $\Gamma^*$-robustness is equivalent to robustness, $\Gamma(\varepsilon)$-robustness is a stronger requirement than robustness. 
However, if $\varepsilon>0$ is sufficiently small, $\Gamma(\varepsilon)$-robustness is equivalent to robustness as shown by the next lemma, which generalizes Lemma~\ref{lemma gamma}.

\begin{proposition}\label{lemma lambda 1}
Let 
\[
\underline\varepsilon=
\min_{\phi\in \Phi:\text{\em $\phi$ is robust}}\, 
\min_{p\in \Delta(\mathcal{X})}\max_{i\in N} \frac{2r_i(\phi,p)-1}{1-r_i(\phi,p)},
\]
which is strictly positive by the definition of robustness. 
Then, for any $\Lambda\subset\Delta(\mathcal{U})$ with $\Lambda^*\subseteq \Lambda \subseteq \Gamma(\underline\varepsilon)$, 
a voting rule $\phi\in\Phi$ is $\Lambda$-robust if and only if it is robust.	
\end{proposition}

\begin{proof}
It is enough to show that robustness implies $\Gamma(\underline\varepsilon)$-robustness. 
Let $\phi\in \Phi$ be a robust rule. 
Fix arbitrary $\lambda \in \Gamma(\underline \varepsilon)$, and let $i\in N$ be the individual who has the maximum responsiveness under $\lambda$.  
Define $\overline{D}_{\lambda} =\max_{x\in \mathcal{X}}E_\lambda[u_i(\delta u_i)-u_i(-\delta u_i)|\delta u=x]$ and $\underline{D}_{\lambda} = \min_{x\in \mathcal{X}}E_\lambda[u_i(\delta u_i)-u_i(-\delta u_i)|\delta u=x]$. 
Note that $\overline{D}_{\lambda}\geq \underline{D}_{\lambda}>0$ and $\overline{D}_{\lambda}/\underline{D}_{\lambda}< 1+\underline{\varepsilon}$ because $\lambda\in\Gamma(\underline{\varepsilon})$. 
Then, 
\begin{align*}
&E_\lambda[u_i(\phi(\delta u))-u_i(-\phi(\delta u))]\\
&= \sum_{x: \phi(x)=x_i}	E_\lambda[u_i(\delta u_i)-u_i(-\delta u_i)|\delta u=x]p_\lambda(x)+
\sum_{x: \phi(x)=- x_i}	E_\lambda[u_i(-\delta u_i)-u_i(\delta u_i)|\delta u=x]p_\lambda(x)\\
&\geq \underline{D}_{\lambda} \sum_{x: \phi(x)=x_i}p_\lambda(x)-\overline{D}_{\lambda}\sum_{x: \phi(x)=- x_i}p_\lambda(x)\\
&> \underline{D}_{\lambda} \Big(
%p_\lambda(\phi(x)=x_i)
r_i(\phi,p_\lambda)
-(1+\underline\varepsilon)
%p_\lambda(\phi(x)=-x_i)
(1-r_i(\phi,p_\lambda))
\Big)\\
%& =\overline{D}_{\lambda} \Big((2+\underline\varepsilon)p_\lambda(\phi(x)=x_i)-1-\underline\varepsilon\Big)
& = \underline{D}_{\lambda} \Big(2
%p_\lambda(\phi(x)=x_i)
r_i(\phi,p_\lambda)
-1-\underline\varepsilon(1-%p_\lambda(\phi(x)=x_i)
r_i(\phi,p_\lambda)
)\Big)\geq 0,
\end{align*}
where the last inequality holds by the  construction of $\underline{\varepsilon}$. 
This implies that the inverse rule is not weakly Pareto-preferred to $\phi$ in terms of expected utility under $\lambda$. Because $\lambda\in \Gamma(\underline{\varepsilon})$ is an arbitrary element, $\phi$ must be $\Gamma(\underline{\varepsilon})$-robust.
\end{proof}

In contrast, if $\varepsilon$ is too large, $\Gamma(\varepsilon)$-robustness is equivalent to the existence of a dictator.

\begin{proposition}\label{lemma lambda 2}
Let $\overline\varepsilon\equiv 2^n-2$. 
Then, for any $\Lambda\subset\Delta(\mathcal{U})$ with $\bigcap_{\varepsilon>\overline\varepsilon}\Gamma(\varepsilon)\subseteq\Lambda$,  a voting rule  $\phi\in \Phi$ is $\Lambda$-robust if and only if there exists a dictator $i\in N$ such that $\phi(x)=x_i$ for all $x\in \mathcal{X}$. 
\end{proposition}
%\begin{proof}
%See Appendix \ref{Proof of Lemma {lemma lambda 2}}. 
%\end{proof}

\begin{proof}
It is enough to show that $\Gamma(\varepsilon)$-robustness implies the existence of a dictator for any  $\varepsilon>\overline\varepsilon$ because $\Gamma(\varepsilon)\supset \Gamma(\varepsilon')$ if $\varepsilon>\varepsilon'$.   
Let $\phi$ be $\Gamma(\varepsilon)$-robust. 
Seeking a contradiction, suppose that there is no dictator.  
For each $x\in \mathcal{X}$, 
let $u^x\in \mathcal{U}$ be such that 
\[
(u_i^x(1),u_i^x(-1))=
\begin{cases}
	(1/(2^n-1),0) &\text{ if }x_i=1=\phi(x),\\
	(1,0) &\text{ if }x_i=1\neq \phi(x),\\
	(0,1/(2^n-1)) &\text{ if }x_i=-1=\phi(x),\\
	(0,1) &\text{ if }x_i=-1\neq \phi(x),
\end{cases}
\]
and let 
$\lambda\in \Delta(\mathcal{U})$ be the uniform distribution over $\{u^x\in \mathcal{U}: x\in \mathcal{X}\}$. 
Note that $\lambda\in \Gamma(\overline\varepsilon)$. 
Because there is no dictator and $\lambda(u^x)=1/\#\mathcal{X}=1/2^n$ for each $x\in \mathcal{X}$, 
the responsiveness of each individual is at most $1-1/2^n$, i.e., 
$r_i(\phi,p_\lambda)
%p_\lambda(\phi(x)=x_i)
\leq 1-1/2^n$ for each $i\in N$. 
Thus, 
\begin{align*}
&E_\lambda[u_i(\phi(\delta u))-u_i(-\phi(\delta u)]\\
%&= \sum_{x\in \mathcal{X}}	E[u_i(\phi(\delta u))-u_i(\phi(-\delta u)|\delta u=x]p(x)\\
&= \sum_{x: \phi(x)=x_i}	E_\lambda[u_i(x_i)-u_i(-x_i)|\delta u=x]p_\lambda(x)+
\sum_{x: \phi(x)=- x_i}	E_\lambda[u_i(-x_i)-u_i(x_i)|\delta u=x]p_\lambda(x)\\
& =
%p_\lambda(\phi(x)=x_i)
r_i(\phi,p_\lambda)
/(2^n-1)-
%p_\lambda(\phi(x)=-x_i) 
(1-r_i(\phi,p_\lambda))
\leq (1-1/2^n)/(2^n-1)-1/2^n=0.
\end{align*}
This implies that the inverse rule is weakly Pareto-preferred to $\phi$ in terms of expected utility under $\lambda$ and that $\phi$ is not $\Gamma(\overline\varepsilon)$-robust. 
\end{proof}